\title{Simultaneous communication in noisy channels}
\author{
Amit Weinstein\thanks{
Blavatnik School of Computer Science, Tel Aviv University, Tel Aviv 69978,
Israel. Email: amitw@tau.ac.il. Research supported in part by an 
ERC advanced grant.}
}

\date{}
\documentclass[11pt]{article}
\usepackage[]{amsfonts,amsmath,amssymb,amsthm}

\newtheorem{theo}{Theorem}
\newtheorem{prop}[theo]{Proposition}
\newtheorem{lemma}[theo]{Lemma}
\newtheorem{coro}[theo]{Corollary}

\newtheorem{definition}{Definition}

\newcommand{\FF}{{\cal F}}
\newcommand{\GG}{{\cal G}}
\def \NN {\mathbb N}

\begin{document}
\maketitle
\begin{abstract}
A sender wishes to broadcast a message of length $n$ over an
alphabet to $r$ users, where each user $i$, $1 \leq i \leq r$ should
be able to receive one of $m_i$ possible messages. The
 broadcast channel
has noise for each of the users (possibly
different noise for different users), who cannot distinguish between
some pairs of letters. The vector $(m_1, m_2, \ldots, m_r)_{(n)}$ is
said to be feasible if length $n$ encoding and decoding schemes
exist enabling every user to decode his message. A rate vector
$(R_1, R_2, \ldots, R_r)$ is feasible if there exists a sequence of
feasible vectors $(m_1, m_2,\ldots, m_r)_{(n)}$ such that $R_i = \lim_{n \mapsto \infty} \frac {\log_2 m_i} {n},
\mbox{for all } i$.

We determine the feasible rate vectors  for several different scenarios
and investigate some of their properties. An
interesting case discussed is when one user can only distinguish
between all the letters in a subset of the alphabet. Tight restrictions on the feasible rate vectors for some specific noise
types for the other users are provided. The simplest non-trivial
cases of two users and alphabet of size three are fully
characterized. To this end a more general previously known result, to
which we sketch an alternative proof, is
used.

This problem generalizes the study of the Shannon capacity of a graph, by considering more than a single user.

\end{abstract}

\section{Introduction}
A sender  has to transmit messages to $r$-users, where the user number $i$ should  be able to receive any
one of $m_i$ messages. To this end, the sender broadcasts a message of length $n$ over an alphabet
$\Sigma$ of size $k$. Each user $i$ has a confusion graph $G_i$ on the set of letters of $\Sigma$,
where two letters $a,b \in \Sigma$ are connected if and only if user $i$ cannot distinguish between
$a$ and $b$. The sender and users can agree on a  (deterministic) coding scheme.  For each possible
values $a_i$ of the messages, $1 \leq a_i \leq m_i$, the scheme should enable the sender to
transmit a string of length $n$ over $\Sigma$ so that each user $i$ will be able to recover $a_i$.
The vector of a scheme for length $n$ messages is $(m_1,m_2, \ldots ,m_r)$. The rate vector of a
sequence of schemes is the limit
$$
\lim_{n \mapsto \infty} (\frac{\log_2 m_1}{n}, \frac{\log_2 m_2}{n}, \ldots ,\frac{\log_2 m_r}{n}),
$$
assuming the limit exists for this sequence. Our objective is to study which vectors  and which
rate vectors are feasible for a given set of confusion  graphs  $G_i$. This seems to be difficult
even for relatively small cases, and reveals some interesting phenomena. Note that this problem
generalizes the problem of computing the \emph{Shannon capacity} of a graph which was first
considered by Shannon in \cite{S}. In the case of a single user (i.e.  a single confusion graph
$G$), the maximum feasible rate is precisely $\log_2 c(G)$ where $c(G)$ denotes the Shannon
capacity of $G$.

Investigating the feasible rate vectors for a given set of confusion graphs raises another interesting
question. What is the maximum capacity of the channel for all users together? The total capacity
can be measured as the sum of rates for each user, which we refer to as the \emph{total rate}.
This sum encapsulates the usability of the channel.

The individual rates we consider in this paper are
sometimes referred to as \emph{private rates}. Similarly, there is
an analogue question about the \emph{common rate}, determining how
much information could we use if we require all users to receive the
same message. This question is outside the scope of our work, but
for completeness we refer the reader to \cite{GKV} for more
information and results.

\subsection{Initial Observations}
The relation between the described problem and the Shannon capacity leads to the following upper bound on
the users' rates and hence for the total rate as well.

\begin{prop}
\label{p-1}
Given $r$ users whose confusion graphs are $G_1, G_2, \ldots, G_r$, a feasible rate vector $(R_1, R_2, \ldots, R_r)$
must satisfy $R_i \leq \log_2 c(G_i)$ for every $1 \leq i \leq r$, hence $\sum_{i=1}^{r} R_i \leq \sum_{i=1}^{r} \log_2 c(G_i)$.
\end{prop}

Although in practice we have several users, their total rate cannot exceed the possible
rate of a single user who shares all their information. Given a set of confusion graphs $G_1,
\ldots, G_r$, let $G = \cap_{i=1}^{r} G_i$ be the graph over the same alphabet $\Sigma$, where $a,b
\in \Sigma$ are connected in $G$ if and only if they are connected in $G_i$ for every $i$. The
confusion graph $G$ represents the information all users have together and therefore can bound
their total rate as follows.

\begin{prop}
\label{p-2} Given $r$ users whose confusion graphs are $G_1, G_2, \ldots, G_r$, any feasible rate
vector $(R_1, R_2, \ldots, R_r)$ must satisfy $\sum_{i=1}^{r} R_i \leq \log_2 c(\cap_{i=1}^{r}
G_i)$.
\end{prop}

 An important simple property of the feasible rate
vectors is \emph{convexity}, which is often referred to as
\emph{time sharing}. Informally, the messages we broadcast can be
shared between two or more broadcasting schemes, where each part
corresponds to a different scheme. This property can be stated
formally as follows.

\begin{prop}
\label{p-3}
Let $G_1, G_2, \ldots, G_r$ be the confusion graphs for $r$ users. Given two feasible rate vectors $\overline{R}, \overline{R}'$ and $\alpha \in [0,1]$,
the rate vector $\alpha \overline{R} + (1-\alpha)\overline{R}'$ is also feasible.
\end{prop}

\begin{proof}
Since both $\overline{R}$ and $\overline{R}'$ are feasible rate
vectors, each has some corresponding encoding scheme. Our new
encoding scheme would be to use the encoding scheme corresponding to
$\overline{R}$ in the first $\alpha n$ coordinates and the one
corresponding to $\overline{R}'$ in the remaining $(1-\alpha)n$
coordinates. The resulting rate vector is precisely $\alpha
\overline{R} + (1-\alpha)\overline{R}'$, as required.
\end{proof}

\begin{coro} \label{c-4}
Let $G_1, G_2, \ldots, G_r$ be the confusion graphs for $r$ users. Given $x_1, x_2, \ldots, x_r \in [0,1]$ so that $\sum_{i=1}^{r} x_i = 1$,
the rate vector $(x_1\cdot  \log_2 c(G_1), \ldots, x_r \cdot  \log_2 c(G_r))$ is feasible.
\end{coro}

\begin{proof}
For every $1 \leq i \leq r$, the rate vector consisting of rate
$\log_2 c(G_i)$ for the $i$'th user and zero rate for all other
users is trivially feasible. The result thus follows by Proposition
\ref{p-3}.
\end{proof}

\subsection{Previous Results}

The problem of simultaneous communication in noisy channels was previously studied
in the theory of broadcast channels (see \cite{Cover} and its references).
For some scenarios, such as the one we will describe shortly,  a full characterization of all feasible rate vectors was found
(see \cite{Marton} and \cite{Pinsker}). This scenario is described here fully as it is used in some of our proofs
and as we also provide a sketch for an alternative proof for it.

Let $\Sigma_k = \{\sigma_1, \sigma_2, \ldots, \sigma_k\}$ be an
alphabet of size $k$ and let $G_1, G_2, \ldots, G_r$ be the
confusion graphs for the $r$ users, where each confusion graph is a
disjoint union of cliques. Given user $i$, one can view his noise as
receiving $y_i = f_i(x)$ whenever $x$ is transmitted, where $
f_i(x): \Sigma_k \mapsto \{1, 2, \ldots, \ell_i\}$ is the index of
the clique which $x$ belongs to and $\ell_i$ is the number of
cliques in the $i$'th user's confusion graph. Note
that we consider isolated vertices as cliques of size one and hence
$f_i$ is well defined up to the order of the cliques.

\begin{definition} \label{d-1}
Given a probability distribution $p$ over $\Sigma_k$ and a subset of the users $I \subseteq \{1, 2,
\ldots, r\}$, let $H_{(p)} (  \{ Y_i \}_{i \in I} )$ be the binary entropy of the random variables
$\{Y_i \}_{i \in I}$ where $Y_i = f_i(X)$ and $X$ is the random variable distributed over
$\Sigma_k$ according to $p$.
\end{definition}

For each subset of the users $I \subseteq \{1, 2, \ldots, r\}$, the alphabet $\Sigma_k$ can be
partitioned into $s \leq k$ disjoint parts ($A_1, \ldots, A_s$) according to what these users
receive ($\{ f_i(x) \}_{i \in I}$). These users cannot distinguish between different letters from
the same part $A_j$, so their joint information when sending a letter $X$ from $\Sigma_k$ according
to the probability distribution $p$ can be computed as
$$
H_{(p)} (  \{ Y_i \}_{i \in I} ) = -\sum_{ 1 \leq j \leq s } \Pr[ X \in A_j] \cdot \log_2 \Pr[X \in
A_j] ~.
$$
Therefore in a sense that will be made precise later, when using only messages in which the letters
are distributed according to some distribution $p$, we expect no subset $I$ of users to have total
rate which exceeds $H_{(p)} ( \{ Y_i \}_{i \in I} )$.

The following theorem, for which we sketch an alternative proof,
provides the full characterization of all feasible rate vectors.

\begin{theo}[\cite{Pinsker}]\label{t-5}
Let $G_1, G_2, \ldots, G_r$ be the confusion graphs for $r$ users over the alphabet $\Sigma_k =
\{\sigma_1, \sigma_2, \ldots, \sigma_k\}$, where each confusion graph is a disjoint union of
cliques. Using the notations and definitions above, a rate vector $(R_1, R_2, \ldots, R_r)$ is
feasible if and only if there exists a probability distribution $p = (p_1, \ldots, p_k)$ over
$\Sigma_k$ so that for every subset $I \subseteq \{1, 2, \ldots, r\}$ of the users,
$$
\sum_{i \in I} R_i \leq H_{(p)} (\{Y_i\}_{i \in I}).
$$
\end{theo}

An interesting special case of the above theorem is the symmetric dense scenario of $r=k$ users,
where the confusion graph $G_i$ of user $i$  is a clique over $\Sigma_k - \{\sigma_i\}$. In other
words, user $i$ only distinguishes the letter $\sigma_i$ from all other letters. When decoding, the
relevant information for user $i$ is only the locations of $\sigma_i$ in the transmitted message.

\begin{coro} \label{c-6}
For every fixed $k \geq 3$, $(\frac {\log_2 k} {k}, \ldots, \frac {\log_2 k}{k})$ is a feasible
rate vector over the alphabet $\Sigma_k = \{\sigma_1, \sigma_2, \ldots, \sigma_k\}$, when each
confusion graph $G_i$ is a clique on $\Sigma_k-\{\sigma_i\}$.
\end{coro}

Corollary \ref{c-6} indicates the possible gain of encoding schemes for several users
simultaneously. The total rate here is $\log_2 k$, whereas using convexity with encoding schemes
for single users cannot exceed a total rate of 1 (as this is the maximum rate for each single user;
the Shannon capacity $c(G_i)$ is precisely 2 for each $1 \leq i \leq k$). However, in some cases
there is no such gain. Several examples are discussed in what follows.

\subsection{Our Results}
For simplicity we omit all floor and ceiling signs whenever these are not crucial.

Let $\Sigma_k = \{\sigma_1, \sigma_2, \ldots, \sigma_k\}$ be an alphabet of size $k$ and let
$\Sigma_d = \{ \sigma_1, \ldots, \sigma_d\}$ denote the set of the first $d$ letters of $\Sigma_k$,
where $2 \leq d \leq k$. Consider the case where user 1 has a confusion graph
$$
G_1 = (\Sigma_k, \{ ab \mid a \in \Sigma_k \wedge b \in \Sigma_k - \Sigma_d\}),
$$
meaning the complete graph over $\Sigma_k$ minus a clique over $\Sigma_d$.  The Shannon capacity of
such graphs is easily shown to be $c(G_1) = d$, hence the maximum rate of user 1 is at most $\log_2
d$. The following results indicate, that for two different confusion graphs of user 2, nothing can
be gained beyond convexity of single user encoding schemes. We need the following definition.

\begin{definition} \label{d-2}
A rate vector is  \emph{optimal} if no user can increase his rate while the other user maintains
the same rate.
\end{definition}

\noindent Note that the total rate of such optimal rate vectors does not necessarily reach the
maximum total rate possible.

\begin{theo} \label{t-7}
In the scenario described above for $2 \leq d \leq k$, when user 2 has the empty confusion graph $G_2 = (\Sigma_k, \emptyset)$,
the rate vectors $(\alpha \log_2 d, (1-\alpha)\log_2 k)$ for $\alpha \in [0, 1]$ are optimal.
\end{theo}

\begin{theo} \label{t-8}
In the scenario described above for $2 \leq d \leq \frac {k+1} {2}$, when user 2 has the complement confusion
graph, that is
$$
G_2 = \overline{G}_1 = (\Sigma_k, \{ ab \mid a,b \in \Sigma_d \}),
$$
the rate vectors $(\alpha \log_2 d, (1-\alpha)\log_2 (k-d+1))$ for $\alpha \in [0, 1]$ are optimal.
\end{theo}

Finally, we provide  a full characterization of all feasible rate vectors for all scenarios containing two users and alphabet of size three
(Propositions \ref{p-15}, \ref{p-16} and \ref{p-17}).

\subsection{Organization}
The rest of the paper is organized as follows. In Section \ref{s-disjoint-cliques} we present a
sketch of an alternative proof to the characterization of the feasible rate vectors for the first
scenario, where each confusion graph is a union of disjoint cliques, (Theorem \ref{t-5}), and
demonstrate how combining encoding schemes for many users can sometimes outperform convexity
(Corollary \ref{c-6}). Section \ref{s-clique-minus-clique} deals with the second family of
confusion graphs described above in which convexity yields the optimal rate vectors (Theorems
\ref{t-7} and \ref{t-8}). Combining these results, one can characterize all feasible rate vectors
for all scenarios involving two users and alphabet of size three. In Section \ref{s-users2-alphabet3}
we elaborate in more details on this analysis. The final Section \ref{s-conclusions} contains some
concluding remarks and open problems.

%%%%%%%%%%%%%%%%%%%%%
%%%%%%%%%%%%%%%%%%%%%
%%%%%%%%%%%%%%%%%%%%%

\section {Unions of disjoint cliques - outperforming convexity}
\label{s-disjoint-cliques}

We consider the case where the confusion graph of each user $i$ is a disjoint union of cliques.
This case is of special interest as a full description of all feasible rate vectors was known and
it is deterministic in the sense that given the transmitted letter, we can transform it deterministically to the
different symbols that each user receives. Moreover, choosing specific confusion graphs, it
demonstrates how the maximum possible total rate can be achieved by combining schemes for many
users (and only this way), even when the confusion graphs are nearly complete.

\subsection{An alternative proof of Theorem \ref{t-5} (sketch) }

Let $G_1, G_2, \ldots, G_r$ be a set of confusion graphs for $r$ users over the alphabet
$\Sigma_k$, where each $G_i$ is a disjoint union of cliques. Given a subset of the users $I
\subseteq [r]$ (where $[r] = \{1, 2, \ldots, r\}$), the following definition and Lemma connects
between the possible number of messages to these users and their binary entropy, when restricted to
a specific distribution of the messages

\begin{definition} \label{d-3}
Given a probability distribution $p$ over $\Sigma_k$ and a subset of the users $I \subseteq [r]$,
let $N_{(p)} (n; \{ Y_i \}_{i \in I})$ denote the number of possible (different) messages for these
users under the restriction that each message is originated in a length $n$ message over $\Sigma_k$
in which $\sigma_i$ appears $p_i n$ times.
\end{definition}

\begin{lemma}
\label{l-9}
$$
\frac {2^{H_{(p)}(\{Y_i\}_{i\in I})n}} {n^{k}} \leq N_{(p)} (n; \{ Y_i \}_{i \in I}) \leq
2^{H_{(p)} ( \{ Y_i \}_{i \in I})n}.
$$
\end{lemma}

\noindent The proof of Lemma \ref{l-9} is a simple consequence of Stirling's formula, which is left to the reader.

\begin{proof}[Upper bound]
Given a scheme of a fixed length $n$ which realizes $(m_1, m_2,
\ldots, m_r)$ messages for the $r$ users, one can divide it into
families according to the number of appearances of the letters
$\sigma_i$ in each message. As there are only $k$ letters and all
the messages are of length $n$, there are at most $n^{k-1}$
different families. Given the probability $p = (p_1, p_2, \ldots,
p_k)$ corresponding to the largest family, Lemma \ref{l-9} therefore
indicates that
$$
\frac {\prod_{i \in [r]} m_i} {n^{k-1}} \leq N_{(p)} (n; \{ Y_i \}_{i \in I}) \leq 2^{H_{(p)} ( \{
Y_i \}_{i \in I})n}.
$$
Recall that $R_i = \lim_{n \mapsto \infty} \log_2 m_i / n$ and hence
\begin{eqnarray*}
\sum_{i \in I} R_i
  & = &  \lim_{n \mapsto \infty} \sum_{i \in I} \frac {\log_2 m_i} {n}
     = \lim_{n \mapsto \infty} \frac {\log_2 \prod_{i \in I} m_i} {n} \\
  & \leq & \lim_{n \mapsto \infty} \frac {H_{(p)}( \{ Y_i \}_{i \in I} ) n + (k-1)\log_2 n }
  {n} \\
  & = &  H_{(p)}( \{ Y_i \}_{i \in I} )
\end{eqnarray*}
as required.
\end{proof}

\begin{remark} Formally, the popular probability distribution
$p=(p_1, p_2, \ldots, p_r)$ depends on $n$ but one can take a
subsequence for which it converges to a single vector $p$,
justifying the computation above. A similar argument
can be used in the following lower bound proof, justifying it for
any probability distribution, even one containing irrational
probabilities. 
\end{remark}

\begin{remark} Similar arguments are sometimes referred to as
\emph{type counting} and could be found, for example, in the book by
Csisz\'{a}r and K\"{o}rner \cite{CK}. They also provide a proof for
a claim similar to Lemma \ref{l-9}.
\end{remark}

\begin{proof}[Lower bound]
Let $p = (p_1, \ldots, p_k)$ be a probability distribution over
$\Sigma_k$ and fix $(R_1, R_2, \ldots, R_r)$ so that for any subset
of the users $I \subseteq [r]$, $\sum_{i \in I} R_i \leq H_{(p)}( \{
Y_i \}_{i \in I} )$. Given some large $n$, we set the number of
messages $m_i$ for each user $i$ to be $m_i = \frac {2^{R_i n}}
{n^{k+2}}$ (which clearly satisfies $\lim_{n \mapsto \infty} \frac
{\log_2 m_i} {n} = \lim_{n \mapsto \infty} \frac {R_i n -
(k+2)\log_2 n} {n} = R_i$). By Lemma \ref{l-9}, for every subset of
the users $I \subseteq [r]$,
\begin{eqnarray}
\label{e-1} \prod_{i \in I} m_i & = &
   \prod_{i \in I} \frac {2^{R_in}} {n^{k+2}} =
   \frac {2^{\sum_{i\in I} R_in}} {n^{(k+2)|I|}} \nonumber \\
  & \leq & \frac {2^{H_{(p)}( \{ Y_i \}_{i \in I} )n}} {n^{k+2}} \leq
   \frac {N_{(p)}(n; \{ Y_i \}_{i \in I} )} {n^2}.
\end{eqnarray}
Our encoding scheme will use only messages in which the letters of
$\Sigma_k$ are distributed according to $p$. For each user $i$ there
are $N_{(p)}(n; \{Y_i\})$ different messages that he can identify.
We randomly divide them into $m_i$ families $\FF_{i, 1}, \ldots,
\FF_{i, m_i}$, where each family represents a different message for
user $i$. When the message $\overline{x} \in \Sigma_k^n$ is
transmitted, user $i$ receives $\overline{y}_i = f_i(\overline{x}) =
f_i(x_1) \cdots f_i(x_n)$ and decodes the message $j$, the single $j
\in [m_i]$ for which $\overline{y}_i \in \FF_{i, j}$.

In order to complete the proof we show the described encoding scheme is valid for some selection of
families  $\FF_{i,j}$. Such a scheme is valid if for every set of messages $\{ j_i \in [m_i] \}_{i
\in [r]}$ there exists a message $\overline{x}$ so that for every user $i \in [r]$, $\overline{y}_i
\in \FF_{i, j_i}$.

Given fixed messages $j_1, j_2, \ldots, j_r$ for the $r$ users,
using the extended Janson inequality (c.f., e.g., \cite{AS}, Chapter
8) and \eqref{e-1} one can show that the probability that there
exists no valid message as required is less than $\frac {1} {k^n}$.
As there are $\prod_{i=1}^{r} m_i \leq k^n$ distinct choices for
messages $j_1, \ldots, j_r$, the assertion of Theorem \ref{t-5}
follows by the union bound.
\end{proof}

\begin{remark}
A full citation of the extended Janson inequality can be found in
Appendix \ref{a-0}, together with a more detailed description of how it is
used here.
\end{remark}

\subsection{Proof of Corollary \ref{c-6}}

This corollary of Theorem \ref{t-5} is for the symmetric dense case where there are $r=k$ users,
and each user $i$ distinguishes a single letter from all other letters. This simple case
demonstrates how the maximum rate can be achieved only by mixing the messages for all the users.
Since each confusion graph $G_i$ is a clique on $\Sigma_k - \{\sigma_i\}$, the Shannon capacity of
this graph is 2 and therefore the maximal rate for any single user is 1. However, the theorem shows
that indeed a total rate of $\log_2 k$ can be achieved, and obviously this is best possible.

Let $k \geq 3$ be fixed.
In order to prove the theorem, a probability distribution $p$ is required so that for every subset of the users
$I \subseteq [r] = [k]$,
\begin{equation}
\label{e-2} \sum_{i \in I} R_i = |I| \frac {\log_2 k} {k} \leq H_{(p)}(\{Y_i\}_{i \in I}).
\end{equation}
Let us consider the uniform probability distribution $p$, $p_i =
1/k$ for every letter $\sigma_i$. When considering all users
together, the random variables $\{Y_i\}_{i \in [r]}$ indicate the
exact letter $x$ that was transmitted. Therefore, since the entropy
is exactly $\log_2 k$, $H_{(p)}(\{Y_i\}_{i \in [r]}) = \log_2 k$
which indeed satisfies \eqref{e-2} as $\sum_{i \in [r]} R_i = r
\frac {\log_2 k } {k} = \log_2 k$.

Given a subset of the users $I \subset [r]$, the random variables $\{Y_i\}_{i \in I}$ indicate
which letter $x$ was transmitted if $x = \sigma_i$ for $i \in I$ or alternatively, that some other
letter was transmitted. Hence its binary entropy satisfies
\begin{eqnarray*}
H_{(p)}(\{Y_i\}_{i \in I}) & = & |I| \frac {\log_2 k} {k} + \frac
{k-|I|} {k} \log_2 \frac {k} {k - |I|} \\
 & > &  |I| \frac {\log_2 k} {k} = \sum_{i \in I} R_i
\end{eqnarray*}
thus \eqref{e-2} holds for every subset of the users $I \subseteq
[r]$, as required. $\qed$

%%%%%%%%%%%%%%%%%%%%%
%%%%%%%%%%%%%%%%%%%%%
%%%%%%%%%%%%%%%%%%%%%

\section {A clique minus a clique - convexity is everything}
\label{s-clique-minus-clique}

We consider the case where the confusion graph $G_1$ of the first
user is the complete graph on $k$ vertices minus a clique on $d$
vertices. We give an upper bound on the rate of the other user for
both the empty confusion graph and for $\overline{G}_1$. In both
cases, the results imply that optimal encoding can be achieved by
convexity, that is nothing can be gained by encoding the messages
together. In order to prove Theorems \ref{t-7} and \ref{t-8}, we
need the following lemmas whose proofs are provided in
Appendix \ref{a-1}.

\begin{lemma}
\label{l-10} Given $a,b \in \NN$ s.t.  $2 \leq b \leq a$ and $x_1 \geq  x_2 \geq  \cdots \geq x_b
\geq 0$,
$$
(a-b+1) a^{\log_b x_b} + \sum_{i \in [b-1]} a^{\log_b x_i} \leq a^{\log_b \sum_{i \in [b]} x_i}.
$$
\end{lemma}

\begin{remark}
Define here  $a^{\log_b 0} = 0$ hence we allow $x_b$ to
be 0.
\end{remark}

\begin{lemma}
\label{l-11} Given $2 \leq d \leq k \in \NN$ and a set $\GG \subseteq \Sigma_k^n$, we define $ \GG'
= \GG \cap \Sigma_d^n$. If $\GG$ is closed under replacing each  $\sigma_i$ with $\sigma_j$ for any
$i > d$ and $j \in [k]$, then either $|\GG| = |\GG'| = 0$ or
$$
\log_k |\GG| \leq \log_d |\GG'|.
$$
\end{lemma}

\begin{lemma}
\label{l-12} Given $2 \leq d \leq k \in \NN$ s.t. $d \leq \frac {k+1} {2}$ and a set $\GG \subseteq
\Sigma_k^n$, define $ \GG' = \GG \cap \Sigma_d^n$ and $\GG'' = \{ f(g_1)f(g_2)\cdots f(g_n) \mid g
\in \GG \} $ where $f(\sigma_i) = \sigma_{\max\{i, d\}}$. If $\GG$ is closed under replacing each
$\sigma_i$ with $\sigma_j$ for any $i > d$ and $j \in [k]$, then either $|\GG'| = |\GG''| = 0$ or
$$
\log_{k-d+1} |\GG''| \leq \log_d |\GG'|.
$$
\end{lemma}

\begin{remark} The restriction of $d$ is required as one can easily find an example where $d=\lceil \frac
{k+1} {2} \rceil$ for which the lemma does not hold (such examples
are given in Appendix \ref{a-2}).
\end{remark}

\begin{proof}[Proof of Theorem \ref{t-7}]
Let $G_1, G_2$ be the confusion graphs as defined in the theorem and
assume the rate of the first user is $\alpha \log_2 d$ for some
$\alpha \in [0, 1]$. The messages used can be divided into disjoint
families $\FF_1,\FF_2, \ldots, \FF_{d^{\alpha n}}$ according to the
message for the first user. Since the first user can only
distinguish between the letters $\sigma_i$ for $i \in [d]$, we can
and will assume each such family $\FF_a$ is closed under replacing
$\sigma_i$ with $\sigma_j$ for $i > d$ and $j \in [k]$.

In order to prove this assumption is valid, it suffices to show user
1 can still distinguish between each of the families after these
replacements. Notice that using this assumption might
result in a different scheme, however, this shows that user 1 hasn't
lost anything from this transition, while user 2 could possibly gain
as we increased the size of each family (and he cannot lose since
one could still use only the original families). Let $\GG_a$ denote
the family $\FF_a$ after replacing $\sigma_i$ with $\sigma_j$ for $i
> d$ and $j \in [k]$. Assume by contradiction that there exist two
families $\FF_a, \FF_b$ and two vectors $v_a \in \GG_a, v_b \in
\GG_b$ so that user 1 can distinguish between $\FF_a$ and $ \FF_b$
but cannot distinguish between $v_a$ and $v_b$. By the definition of
$G_1$, for every coordinate $i \in [n]$, either $v_a[i] \not \in
\Sigma_d$ or $v_b[i] \not \in \Sigma_d$ or $v_a[i] = v_b[i] \in
\Sigma_d$ as otherwise user 1 would be able to distinguish between
them (here $v_x[i]$ denotes the $i$'th letter in the vector $v_x$).
Since $v_a \in \GG_a$ and $v_b \in \GG_b$, we know there exists $u_a
\in \FF_a$ and $u_b \in \FF_b$ from which $v_a$ and $v_b$ can be
derived by the replacements above. Therefore, for every coordinate
$i \in [n]$, either $u_a[i] \not \in \Sigma_d$ or $u_b[i] \not \in
\Sigma_d$ or $u_a[i] = u_b[i] \in \Sigma_d$. However, this is in
contradiction to the fact that user 1 was able to distinguish
between $\FF_a$ and $\FF_b$ as he cannot distinguish between $u_a$
and $u_b$.

Define $\FF_a' = \FF_a \cap\Sigma_d^n$ for every $\FF_a$. Since
these families are pairwise disjoint, by an averaging argument there
exists some message $a$ for which $|\FF_a'| \leq d^{-\alpha n} \cdot
d^n = d^{(1-\alpha)n}$. By Lemma \ref{l-11}, for this specific
message $a$, $|\FF_a| \leq k^{(1-\alpha)n}$ which implies the rate
of the second user is at most $(1-\alpha)\log_2 k$.
\end{proof}

\begin{coro}
\label{c-13} For the confusion graph $G_1$ as above and any confusion graphs $G_2, \ldots, G_r$, a
feasible rate vector $(\alpha \log_2 d, R_2, \ldots, R_r)$ for $\alpha \in [0,1]$ must satisfy
$\sum_{i=2}^{r} R_i \leq (1-\alpha) \log_2 k$.
\end{coro}

\begin{proof}[Proof of Theorem \ref{t-8}]
Let $G_1, G_2$ be the confusion graphs as defined
in the theorem. Note that the Shannon capacity $c(G_2)$ is precisely
$k-d+1$, hence these rate vectors are feasible by Corollary
\ref{c-4} (as is also easy to see directly). Assume the rate of the
first user is $\alpha \log_2 d$ for some $\alpha \in [0, 1]$. The
messages used can be divided into disjoint families $\FF_1,\FF_2,
\ldots, \FF_{d^{\alpha n}}$ according to the message for the first
user. Since the first user can only distinguish between the letters
$\sigma_i$ for $i \in [d]$, we can and will assume, as in the proof
of Theorem \ref{t-7}, that each such family $\FF_a$ is closed under
replacing $\sigma_i$ with $\sigma_j$ for $i > d$ and $j \in [k]$.

Define $\FF_a' = \FF_a \cap \Sigma_d^n$ for every $\FF_a$. Since these families are pairwise
disjoint, by an averaging argument there exists some message $a$ for which $|\FF_a'| \leq
d^{-\alpha n} \cdot d^n = d^{(1-\alpha)n}$. Given the first user should receive the message $a$,
the second user has at most $|\FF_a''| = |\{   f(g_1)f(g_2)\cdots f(g_n) \mid g \in \FF_a\} |$
different messages where $f(\sigma_i) = \sigma_{\max\{i, d\}}$ (as the second user can only
distinguish the locations of $\sigma_i$ for $i \in [k] - [d]$ and all other letters are
indistinguishable for him). By Lemma \ref{l-12}, for this specific message $a$, $|\FF_a''| \leq
(k-d+1)^{(1-\alpha)n}$ which implies the rate of the second user
 is at most $(1-\alpha)\log_2 (k-d+1)$.
\end{proof}

\begin{coro}
\label{c-14} For the confusion graph $G_1$ as above and any confusion graphs $G_2, \ldots, G_r
\supseteq   \overline{G}_1 $, a feasible rate vector $(\alpha \log_2 d, R_2, \ldots, R_r)$ for
$\alpha \in [0,1]$ must satisfy $\sum_{i=2}^{r} R_i \leq (1-\alpha) \log_2 (k-d+1)$.
\end{coro}

%%%%%%%%%%%%%%%%%%%%%
%%%%%%%%%%%%%%%%%%%%%
%%%%%%%%%%%%%%%%%%%%%

\section {Two users, three letters - the complete story}
\label{s-users2-alphabet3}

Two users and three letters is the smallest possible example of non-trivial scenario. Having only
two letters result in each user either knowing everything or knowing nothing and obviously having a
single user coincides with the Shannon capacity question. These smallest scenarios however already
contain some interesting cases which we analyze using the previous results.

Let $\Sigma = \{\sigma_0, \sigma_1, \sigma_2\}$ be our alphabet and let $G_1, G_2$ be the confusion
graphs of the two users correspondingly. If one of the confusion graphs is the complete graph,
again it coincides with the Shannon capacity of a single graph (for the non-complete confusion
graph). As stated earlier, there is a strong connection between the feasible rate vectors and the
Shannon capacity of graphs. In the cases we are about to analyze, we use the fact that the Shannon
capacity of every graph on 3 vertices which is neither the empty graph nor the clique, is precisely
2 (each such graph is perfect, hence its Shannon capacity equals its independence number). By the
symmetry between the users and the letters in the alphabet, it suffices to discuss only subset of
the possible confusion graphs.

\subsection {Confusion graph with two edges}
In this subsection we show that when the first confusion graph has two edges, then no scheme can
outperform what follows from convexity.
\begin{prop}
\label{p-15} Let $G_1 = (\Sigma, \{\sigma_0\sigma_1, \sigma_0\sigma_2\})$, meaning the first user
only distinguishes between the letters $\sigma_1$ and $\sigma_2$. For every $G_2$, the optimal rate
vectors are given by Corollary \ref{c-4}, i.e.
$$
(\alpha \cdot \log_2 c(G_1), (1-\alpha) \cdot \log_2 c(G_2))
$$
for $\alpha \in [0,1]$ where $\log_2 c(G_1) = 1$.
\end{prop}

\begin{remark} Note that this matches the case of a clique minus a clique for the parameters $k=3$ and
$d=2$ as denoted in previous sections, but here we do not limit the
confusion graph of the second user.
\end{remark}

\begin{proof}
The proof is divided into two parts, according to the intersection
between the edges of $G_1$ and $G_2$. In the first case where there
is a non-empty intersection, the bound given by Proposition
\ref{p-2} yields a maximum total rate of $\log_2 c(G_1 \cap G_2) =
1$. Therefore, one could not hope for finding a feasible rate vector
which is not of this form (assuming $G_2$ is not the complete graph,
these are all optimal rate vectors as they have a total rate of 1).

Let us assume there is no intersection between the two confusion
graphs, meaning either $G_2$ is the empty graph or $G_2 =
\overline{G}_1$. These cases match Theorems \ref{t-7} and \ref{t-8}
respectively which indeed yield the desired result.
\end{proof}

\subsection {The first confusion graph has a single edge}

Throughout this section we denote $H$ as the binary entropy function where given some probability
distribution $p_1, p_2, \ldots, p_k, q$ (where $q = 1-\sum_{i=1}^{k} p_i$),
\begin{eqnarray*}
H(p_1, p_2, \ldots, p_k, q) & = & H(p_1, p_2, \ldots, p_k) \\
& = & -\sum_{i=1}^{k} p_i \log_2 p_i - q \log_2 q.
\end{eqnarray*}

\begin{prop}
\label{p-16} Let $G_1 = (\Sigma, \{\sigma_0\sigma_1\})$ and $G_2 = (\Sigma,\{\sigma_0\sigma_2\})$.
The following rate vectors are optimal:
 \begin{itemize}
  \item $(R_1, H(R_1))$ for $R_1 \in [1/2, 2/3]$.
  \item $(R_1, \log_2 3 - R_1)$ for $R_1 \in [2/3, \log_2 3 - 2/3]$.
  \item $(H(R_2), R_2)$ for $R_2 \in [1/2, 2/3]$.
 \end{itemize}
\end{prop}

\begin{remark} By the proposition a rate vector $(R_1, R_2)$ in the above case is feasible if and only if
\begin{enumerate}
    \item $R_1 \in [0, 1/2]$ and $R_2 \in [0, 1]$ or
    \item $R_1 \in [1/2, 2/3]$ and $R_2 \in [0, H(R_1)]$ or
    \item $R_1 \in [2/3, \log_2 3 - 2/3]$ and $R_2 \in [0, \log_2 3 - R_1]$ or
    \item $R_1 \in [\log_2 3 - 2/3, 1]$ and $R_2 \in [0, H^{-1}(R_1)]$.
\end{enumerate}
\end{remark}
%This is depicted in Figure \ref{}.

\begin{proof} The scenario described above is a special case of Theorem \ref{t-5}. Given a probability
distribution $p = (p_0, p_1, p_2)$, $y_1$ is distributed $(p_0+p_1,p_2)$, $y_2$ is distributed
$(p_0 + p_2, p_1)$ and $\{y_1, y_2\}$ is distributed according to $p$.

The uniform distribution $p = (\frac 1 3, \frac 1 3, \frac 1 3)$ yields that the rate vectors
$(R_1, R_2)$ are feasible if $R_1 + R_2 \leq \log_2 3$ and each $R_i \leq H(2/3) = \log_2 3 - 2/3$.
This matches the second case described in the theorem, which is obviously optimal as one
cannot hope to exceed a total rate of $\log_2 3$.

By symmetry, it suffices to analyze the first case of the theorem in order to complete the proof.
Setting $p_1$ to be some probability smaller than half bounds the rate of the second user by $R_2
\leq H(p_1) = H(1-p_1)$. Moreover, the total rate $R_1 + R_2$ is bounded by $H(\frac {1-p_1} {2},
\frac {1-p_1}{2}, p_1) = H(p_1) + (1-p_1)$. This shows that the rate vectors $(R_1, H(R_1))$ are
feasible as $R_1 = 1-p_1 \leq H(\frac {1-p_1} {2})$ for $p_1 \in [0, 1/2]$ (indeed equality holds
for $p_1 = 0$ and since $H'(x) = \log_2(1-x) - \log_2 x < 2$ for $x \in [1/4, 1/2]$, or
equivalently $H'(\frac {1-p_1} {2}) > -1$ for $p_1 \in [0, 1/2]$, this holds for every $p_1 \in [0,
1/2]$ as well). Moreover, these rate vectors are also optimal as the bound for the total rate
$H(p_1) + (1-p_1)$ decreases while $p_1$ increases in the section $[1/3, 1/2]$ (using $H'(x) < 1$
for $x \in [1/3, 1/2]$).
\end{proof}

\begin{remark} Although the rate vector $(1, 1/2)$ is feasible, the vector $(2^n, 2)$ is not
feasible. If user $1$ needs to be able to receive $2^n$ distinct
messages, one of them has to be encoded by $(\sigma_2,\sigma_2,
\cdots ,\sigma_2)$. But this has to be transmitted independently of
the message of the second user, showing there is no
$(2^n,2)$-scheme.
\end{remark}

\begin{prop}
\label{p-17} Let $G_1 = (\Sigma, \{\sigma_0\sigma_1\})$ and $G_2$ be the empty graph. The following
rate vectors are optimal:
 \begin{itemize}
  \item $(R_1, \log_2 3 - R_1)$ for $R_1 \in [0, H(2/3)]$.
  \item $(H(R_2), R_2)$ for $R_2 \in [1/2, 2/3]$.
 \end{itemize}
\end{prop}

\begin{remark} By the proposition a rate vector $(R_1, R_2)$ in the above case is feasible if and only if
\begin{enumerate}
    \item $R_1 \in [0, \log_2 3 - 2/3]$ and $R_2 \in [0, \log_2 3 - R_1]$ or
    \item $R_1 \in [\log_2 3 - 2/3, 1]$ and $R_2 \in [0, H^{-1}(R_1)]$.
\end{enumerate}
%This is depicted in Figure \ref{}.
\end{remark}

\begin{proof}
Our problem is monotone in the following sense. Removing
an edge from one of the confusion graphs can only improve the
feasible rate vectors. Therefore in our case, the lower bound of
$(H(R_2), R_2)$ for $R_2 \in [1/2, 2/3]$ we achieved when both
confusion graphs had a single edge can be applied here. Showing
these rate vectors are also optimal will complete the proof as we
have the trivial upper bound of $\log_2 3$ on the total rate, and by
combining convexity with the fact that the rate vectors
$(H(2/3),2/3)$ and $(0, \log_2 3)$ are feasible, we conclude that
all other required rate vectors are achieved.

Our problem is a special case of Theorem $\ref{t-5}$. Note that in
the proof of Proposition $\ref{p-16}$ we showed an
upper bound for the rates $(R_1, H(R_1))$ which only depended on the
second user. Assuming the second user has rate of $H(R_1)$ already
bounds the total rate by $H(R_1) + R_1$. Similarly in our case,
assuming the first user has rate of $H(R_2)$ for $R_2 \in [1/2,
2/3]$ bounds the total rate of the two users together by $H(R_2) +
R_2$.
\end{proof}

%%%%%%%%%%%%%%%%%%%%%
%%%%%%%%%%%%%%%%%%%%%
%%%%%%%%%%%%%%%%%%%%%

\section {Conclusions and open problems}
\label{s-conclusions}
In this work we have studied the notion of simultaneous communication in a
noisy channel where the channel's noise may differ for each of the users. The goal is to find, for
a given set of confusion graphs which represent the noise for each of the users, which rate vectors
(or alternatively vectors) are feasible. As in the Shannon capacity of a channel, we care about the
average rate per letter when the length of the messages tends to infinity.

Our work demonstrates basic lower and upper bounds for the general case. A simple yet useful tool
in understanding the feasible rate vectors is the convexity property which is described in
Proposition \ref{p-3}. We saw several examples where convexity and basic encoding schemes (derived
from the Shannon capacity of the confusion graphs) are optimal. On the other hand, there are
examples where much more can be gained by mixing the encoding for several users. The case in which
every graph is a disjoint union of cliques is fully understood, and so is the case of 2 users and
alphabet of size 3. Many other cases remain open.

\bigskip\noindent
We conclude with several open problems it would be
interesting to solve.
\begin{itemize}
 \item
The lower and upper bounds for the maximum total rate given in this paper apply combinatorial and
probabilistic techniques. It would be interesting to find stronger bounds which possibly extend the
algebraic and geometric bounds known for the Shannon capacity, such as the bounds given by Lov\'asz
in \cite{L}, Hamers \cite{H} or Alon \cite{A}.

 \item
In the non-symmetric case where we have a user whose confusion graph is a clique over $k$ letters
minus a clique over $d$ letters and the other user's confusion graph is its complement, it would be
interesting to know if Theorem \ref{t-8} still holds for $d > \frac{k+1}{2}$. Since Lemma
\ref{l-12} does not hold for such $d$, a different approach must be used.

 \item
It seems interesting to study graphs $G$ for which the maximum total rate is as small as possible
using $G$ and $\overline{G}$ as the two confusion graphs for two users.
In such a case, the upper bound of the Shannon capacity of the intersection (Proposition \ref{p-2})
does not help as the two graphs are disjoint. However, by Theorem 1.1 of
\cite{A}, we know there exist graphs $G$ on $k$ vertices for which both $c(G)$ and
$c(\overline{G})$ are at most $e^{O(\sqrt{\log k \log \log k})}$.
For such a graph, the upper bound in Proposition \ref{p-1} of the maximum total rate yields
$O(\sqrt{\log k \log \log k})$ which is far less than the trivial upper bound of $\log_2 k$.

 \item
Most of the encoding schemes considered in this paper use randomness and therefore are not given
explicitly. As a result, the encoding and decoding schemes are not efficient. Finding explicit and
efficient encoding and decoding schemes for the scenarios described in the paper remains open.
\end{itemize}

\section*{Acknowledgment}
 I am grateful to Alon Orlitsky and Ofer Shayevitz  for helpful
discussions. I am especially thankful to Noga Alon for his dedication, guidance and support throughout
this research.

\appendix
\section{The extended Janson inequality and its application in Theorem \ref{t-5}}
\label{a-0}

Below is the full citation of the extended Janson inequality,
followed by its application in the lower bound proof of Theorem
\ref{t-5}.

\begin{theo}[Janson]
\label{t-janson} Let $S$ be a set, for each $s \in S$ let $p_s$ be a
real $0 \leq p_s \leq 1$. Let $R$ be a random subset of $S$ obtained
by selecting each $s \in S$, randomly and independently, to lie in
$R$ with probability $p_s$. Let $A_i, i \in I$ be a family of
subsets of $S$. For each $i \in I$, let $B_i$ be the event that $A_i
\subset R$. Let $\mu=\sum_{i \in I} Prob [B_i]$ be the expected
number of events $B_i$ that occur. Define $\Delta=\sum Prob[B_i
\wedge B_j]$, where the sum ranges over all ordered pairs $i,j \in
I, i \neq j$ such that $A_i \cap A_j \neq \emptyset$. Then the
probability that none of the events $B_i$ occurs is at most
$e^{-\mu+\Delta/2}$. If the further assumption that $\Delta > \mu$
holds, this probability can also be bounded by $e^{-\mu^2/2\Delta}$.
\end{theo}

Given the process defined in the lower bound proof of Theorem
\ref{t-5}, we fix messages $j_1, j_2, \ldots, j_r$ for the $r$
users. Consider the $r$-uniform $r$-bipartite hypergraph whose
classes of vertices are $\FF_{i, j_i}$, $1 \leq i \leq r$. Each edge
represents a consistent message, i.e. for every message
$\overline{x} \in \Sigma_k^n$ there exists an edge $\{
\overline{y}_i \in \FF_{i, j_i} \}_{i \in [r]}$ if indeed for every
$i \in [r]$, $\overline{y}_i = f_i(\overline{x}) \in \FF_{i, j_i}$.
Existence of some edge in the hypergraph indicates that this set of
messages can be transmitted as required.

Using the notations of Theorem \ref{t-janson}, our set $S$ is the
union of all possible messages $f_i(\overline{x})$ for each user $i$
and for $\overline{x} \in \Sigma_k^n$ which is distributed according
to $p$. The probability of each element is $1/m_i$ for the relevant
user $i$. The sets $A_i$ represent all the consistent messages
$\{f_i(\overline{x}) \}_{i \in [r]}$ where again, $\overline{x} \in
\Sigma_k^n$ and is distributed according to $p$.

By Theorem \ref{t-janson} the probability that there exists no valid
message as required is at most $e^{-\mu^2/2\Delta}$. One can now
verify that indeed, as defined here, $\mu$ and $\Delta$ satisfy
$e^{-\mu^2/2\Delta} \leq e^{-n^2/2^{k+1}} < \frac {1} {k^n}$. 

\section{Proofs of Lemmas \ref{l-10}, \ref{l-11} and \ref{l-12}}
\label{a-1}

\begin{proof}[Proof of Lemma \ref{l-10}]
When $x_1 = x_2 = \cdots = x_b$ equality holds as
\begin{eqnarray*}
& (a-b+1)a^{\log_b x_b} + \sum_{i \in [b-1]} a^{\log_b x_i} = a
\cdot a^{\log_b x_b} \\ & = a^{1 + \log_b x_b} = a^{\log_b b \cdot
x_b} = a^{\log_b \sum_{i \in [b]} x_i}.
\end{eqnarray*}
In order to complete the proof, it suffices to show that the partial
derivatives $\frac {\partial} {\partial x_i}$ are smaller on the
left hand side than those on the right hand side for any $i \in
[b-1]$, regardless of the values $\{x_i\}$. Given a fixed $i \in
[b-1]$, the derivative of the left hand side is
$\frac{\partial}{\partial x_i} a^{\log_b x_i} =
\frac{\partial}{\partial x_i} x_i^{\log_b a} = \log_b a \cdot x_i
^{\log_b a - 1}$. On the other hand, the derivative of the right
hand side is $ \log_b a \cdot (\sum_{i \in [b]} x_i)^{\log_b a - 1}$
which is at least as big.
\end{proof}

\begin{proof}[Proof of Lemma \ref{l-11}]
We apply induction on $n$. For $n=1$, if $\GG' = \GG$ the inequality
holds as $k \geq d$ (or both sets are empty). Otherwise there exists
$\sigma_i \in \GG$ for $i > d$, hence $|\GG| = |\Sigma| = k$ and
$|\GG'| = d$ for which equality holds.

Assuming the lemma holds for any $n' < n$ we prove it for $n$.
Define $\GG_i = \{ g_1g_2\ldots g_{n-1} \mid g \in \GG \wedge g_{n}
= \sigma_i \}$ and $\GG_i' = \GG_i \cap \Sigma_d^{n-1}$ for every $i
\in [k]$. Note that $\GG_i' \subseteq \GG_j'$ and hence $|\GG_i'|
\leq |\GG_j'|$ for every $i > d$ and $j \in [d]$ (since $\GG$ is
closed under replacing $\sigma_i$ with $\sigma_j$ for $i > d$ and $j
\in [k]$). In particular, this is true for $m \in [d]$ so that
$|\GG_m'| = \min_{j \in [d]} |\GG_j'|$. Therefore, by the induction
hypothesis and Lemma \ref{l-10} with $a = k$ and $b = d$,
\begin{eqnarray*}
 |\GG|  & = &  \sum_{i \in [k]} |\GG_i| \\
  & \leq &  \sum_{i \in [k]} k^{\log_d |\GG_i'|}  \\
  & \leq & (k-d) k^{\log_d |\GG_m'|} + \sum_{i\in[d]} k^{\log_d |\GG_i'|}  \\
  & \leq & (k-d+1) k^{\log_d |\GG_m'|} + \sum_{i \in [d] - \{m\}} k^{\log_d |\GG_i'|} \\
  & \leq & k^{\log_d \sum_{i \in [d]} |\GG_i'|} = k^{\log_d |\GG'|}
\end{eqnarray*}
completing the proof.
\end{proof}

\begin{proof}[Proof of Lemma \ref{l-12}]
Again we apply induction on $n$. For $n=1$, if $\GG' = \GG$ we have
no $\sigma_i \in \GG$ for $i > d$ and the inequality holds as
$|\GG''| = 1$ or both sets are empty. Otherwise there exists
$\sigma_i \in \GG$ for $i > d$, hence $|\GG''| = k - d + 1$ and
$|\GG'| = d$ for which equality holds.

Assuming the lemma holds for any $n' < n$ we prove it for $n$.
Extending the previous definitions of $\GG_i$ and $\GG_i'$, let
$\GG_i'' = \{ f(g_1)f(g_2)\cdots f(g_{n-1}) \mid g \in \GG_i \}$ for
every $i \in [k]$. Note that $\GG_i'' \subseteq \GG_j'' $ and hence
$|\GG_i''| \leq |\GG_j''|$ for every $i > d$ and $j \in [k]$
 (since $\GG$ is closed under replacing $\sigma_i$ with $\sigma_j$ for $i > d$ and $j \in [k]$).
Similarly, $|\GG_i''| \leq |\cap_{j \in [d]} \GG_j'' |$ for all
$i>d$. Therefore,
 \begin{eqnarray*}
 | \cup_{j\in[d]} \GG_j'' |
  & \leq & \sum_{j \in [d] } |\GG_j''|  - (d-1) \cdot |\cap_{j \in [d]} \GG_j'' | \\
  & \leq & \sum_{j \in [d] } |\GG_j''|  - (d-1) \cdot \max_{i \in [k]-[d]} |\GG_i'' | \\
  & \leq & \sum_{j \in [d] } |\GG_j''| -\frac {d-1} {k-d} \sum_{i \in [k]-[d] } |\GG_i''|.
\end{eqnarray*}
As before, $|\GG_i'| \leq |\GG_m'|$ for every $i > d$ and $m \in
[d]$ so that $|\GG_m'| = \min_{j \in [d]} |\GG_j'|$. Therefore, by
the induction hypothesis
\begin{eqnarray*}
 \sum_{i \in [k]-[d] } |\GG_i''|
    & \leq & \sum_{i \in [k]-[d] } (k-d+1)^{\log_d |\GG_i'|} \\
    & \leq & (k-d) (k-d+1)^{\log_d |\GG_m'|}.
\end{eqnarray*}
By Lemma \ref{l-10} with $a = k-d+1$ and $b = d$ (which indeed
satisfies $2 \leq b \leq a$ as $d \leq (k+1)/2$),
\begin{eqnarray*}
 |\GG''|  & = &  \sum_{i \in [k] - [d]} |\GG_i''| + | \cup_{j\in[d]} \GG_j'' | \\
  & \leq & \frac {(k-d)  - (d-1)} {k-d} \sum_{i \in [k]-[d] } |\GG_i''| +  \sum_{j \in [d] } |\GG_j''|  \\
  & \leq &  \frac {k-2d+1} {k-d}  (k-d) (k-d+1)^{\log_d |\GG_m'|}  \\
  & &  + \sum_{j\in[d]} (k-d+1)^{\log_d |\GG_j'|} \\
  & = &  (k-2d+2)(k-d+1)^{\log_d |\GG_m'|}   \\
  & &  + \sum_{j\in[d] - \{m\}} (k-d+1)^{\log_d |\GG_j'|}  \\
  & \leq &  (k-d+1)^{\log_d \sum_{j \in [d]} |\GG_j'|} = (k-d+1)^{\log_d |\GG'|}
\end{eqnarray*}
completing the proof.
\end{proof}

\section{An example in which Lemma \ref{l-12} does not hold when $d > \frac {k+1}{2}$}
\label{a-2} Let $d = 3$, $k = 4$ and $n = 2$ where indeed $d > \frac
{k+1}{2} = 2.5$. Define
\begin{eqnarray*}
\GG & = & \{(\sigma_1, \sigma_1), (\sigma_1, \sigma_2), (\sigma_1, \sigma_3), (\sigma_1, \sigma_4), \\
  & & \: \ (\sigma_2, \sigma_1), (\sigma_3, \sigma_1), (\sigma_4, \sigma_1) \}
\end{eqnarray*}
which can also be viewed as $\{ (\sigma_1, \sigma_4),
(\sigma_4, \sigma_1) \}$ after replacing $\sigma_4$ with every
$\sigma \in \Sigma_4$ (as $\GG$ has to be closed under these
replacements). By the definitions of the lemma,
\begin{eqnarray*}
\GG' & = & \{
  (\sigma_1, \sigma_1), (\sigma_1, \sigma_2), (\sigma_1, \sigma_3),
  (\sigma_2, \sigma_1), (\sigma_3, \sigma_1) \} , \\
\GG'' & = & \{ (\sigma_3, \sigma_3), (\sigma_3, \sigma_4), (\sigma_4, \sigma_3) \}
\end{eqnarray*}
and therefore the lemma does not hold as
$$
\log_{k-d+1} | \GG'' | = \log_2 3 > \log_3 5 = \log_{d} | \GG'|.
$$

The same example can be used with larger parameters, for instance with $k = 100$ and $d=51$
(which is the minimal $d$ for which $d > \frac {k+1}{2}$). With these parameters,
$|\GG'| = 2d -1 = 101$ and $|\GG''| = 2(k-d+1) - 1 = 99$ and indeed
$\log_{k-d+1} | \GG'' |  = \log_{50} 99 > \log_{51} 101 =  \log_{d} | \GG'| $.

\end{document}